\documentclass[letterpaper, 10 pt, conference]{ieeeconf}

\usepackage[utf8]{inputenc}

\usepackage{amssymb}       
\usepackage{bm}
\usepackage{bbm,marginnote,tcolorbox,xcolor}
\usepackage{mathtools}
\usepackage[colorinlistoftodos]{todonotes}
\usepackage[english]{babel}

\IEEEoverridecommandlockouts 
\usepackage{amsmath,amsthm}
\usepackage{graphicx}

\usepackage[utf8]{inputenc} 
\usepackage[T1]{fontenc}    
\usepackage{hyperref}       
\usepackage{url}            
\usepackage{booktabs}       
\usepackage{amsfonts}       
\usepackage{nicefrac}       
\usepackage{microtype}      
\usepackage{xcolor}         
\usepackage{cite}
\usepackage{caption}
\usepackage{subcaption,gensymb}

\usepackage{hyperref}
\hypersetup{colorlinks,linkcolor={blue},citecolor={blue},urlcolor={purple}}  

\graphicspath{./Figures/}

\newcommand{\ignore}[1]{}  

\DeclareMathOperator*{\minimize}{minimize}

\newtheorem{lemma}{Lemma}
\newtheorem{theorem}{Theorem}
\newtheorem{proposition}{Proposition}
\newtheorem{assumption}{Assumption}
\newtheorem{definition}{Definition}

\theoremstyle{definition}
\newtheorem{remark}{Remark}

\newcommand{\dimState}{\ensuremath{n}}
\newcommand{\dimControl}{\ensuremath{p}}

\newcommand{\noise}{\ensuremath{\omega}}

\newcommand{\safeSet}{\ensuremath{S}}
\newcommand{\safeProb}{\ensuremath{V}}
\newcommand{\safeProbOpt}{\ensuremath{V^\star}}

\newcommand{\Rb}{\ensuremath{\mathbb{R}}}

\newcommand{\longthmtitle}[1]{\mbox{}\textup{\textsl{(#1):}}}

\newcommand{\CME}{\ensuremath{\psi}}

\newcommand{\Pbhat}{\widehat{\Pb}}

\newcommand{\Pb}{\ensuremath{\mathbb{P}}}

\newcommand{\PP}{\ensuremath{\mathcal{P}}}
\newcommand{\XX}{\ensuremath{\mathcal{X}}}
\newcommand{\ZZ}{\ensuremath{\mathcal{Z}}}
\newcommand{\YY}{\ensuremath{\mathcal{Y}}}
\newcommand{\MM}{\ensuremath{\mathcal{M}}}
\newcommand{\HH}{\ensuremath{\mathcal{H}}}

\newcommand{\Eb}{\ensuremath{\mathbb{E}}}
\newcommand{\Qb}{\ensuremath{\mathbb{Q}}}

\newcommand{\innerRKHS}[2]{\ensuremath{\langle #1, #2 \rangle}_{\mathcal{H}_{\mathcal{X}}}}

\graphicspath{./Figures/}

\definecolor{myBlue}{RGB}{0,0,255}
\newif\ifshowmodifications

\newcommand{\changesroundOne}[1]{\ifshowmodifications\textcolor{myBlue}{#1}\else #1\fi}

\title{Distributionally Robust Optimal and Safe Control of Stochastic Systems via Kernel Conditional Mean Embedding [extended version]}

\author{Licio Romao, Ashish R. Hota, and Alessandro Abate%
\thanks{L. Romao and A. Abate are with the Department of Computer Science, Oxford University, UK. Email adresses: \texttt{\{licio.romao,aabate\}@cs.ox.ac.uk}. A. R. Hota is with the Department of Electrical Engineering, Indian Institute of Technology, Kharagpur, India. Email: \texttt{ahota@ee.iitkgp.ac.in}.}}

\begin{document}

\maketitle

\begin{abstract}
We present a novel distributionally robust framework for dynamic programming that uses kernel methods to design feedback control policies. \changesroundOne{Specifically, we leverage kernel mean embedding to map the transition probabilities governing the state evolution into an associated repreducing kernel Hilbert space. Our key idea lies in combining conditional mean embedding with the maximum mean discrepancy distance to construct an ambiguity set, and then design a robust control policy using techniques from distributionally robust optimization. The main theoretical contribution of this paper is to leverage functional analytic tools to prove that optimal policies for this infinite-dimensional min-max problem are Markovian and deterministic.  Additionally, we discuss approximation schemes based on state and input discretization to make the approach computationally tractable. To validate the theoretical findings, we conduct an experiment on safe control for thermostatically controlled loads (TCL).} 

\end{abstract}

\section{Introduction}

\changesroundOne{We focus on discrete-time stochastic control problems, where states evolve according to an underlying stochastic transition kernel. The main objective is to design a feedback control policy that either minimizes an objective function or that maximizes the probability of satisfying temporal properties, such as safety or reach-avoid specifications \cite{Abate2008,Summers2010}.} 


\changesroundOne{Drawing inspirations from recent advancements in the data-driven control literature \cite{Persis2020,Berberich2021}, which advocates for techniques that enable the design of feedback controllers based on available data, we introduce a novel design approach based on dynamic programming that leverages available trajectories of the system. To address sampling errors resulting from finite data, we employ techniques from ``distributionally robust'' optimization and control \cite{Shapiro2017,MohajerinEsfahani2018,HCL19,ALCD22}, which have gained significant attention in the community. These techniques induce robustness in the space of probability measures through the creation of the so-called ambiguity sets. In other words, distributionally robust techniques compute the worst-case expected value of a function, when the underlying measure belongs to such an ambiguity set. In this paper, we relate these ideas to the class of min-max control problems initially investigated in the seminal work
\cite{GonzalezTrejo2002} and further explored in \cite{Ding2013}.} 

\changesroundOne{A key feature of our approach is the use of kernel methods \cite{Muandet2017,ZJDS21}, specifically conditional mean embedding and Maximum Mean Discrepancy (MMD) \cite{SGSS07}. These methods are employed to embed probability distributions into the associated reproducing kernel Hilbert space (RKHS) and measure distance between two probabilities distributions by using MMD. We create an ambiguity set represented as  a ball in the space of probability measures, while the mean embedding provides an estimate of the center of the generated ambiguity set. This allows us to reason about uncertainty and ensure robustness in the decision-making process.}

\changesroundOne{The construction of a min-max control problem through the ambiguity set formulation is relevant, especially in data-limited scenarios where the empirical estimate may not be rich enough to approximate the true conditional mean embedding. While
conditional mean embedding and its empirical estimate have
been applied in the context of Bayesian inference \cite{Fukumizu2013},
dynamical systems \cite{BGG13} and more recently for reachability
analysis \cite{Thorpe2019,Thorpe_2022}, we are not aware of any work that
leverages this framework for control synthesis using distributionally robust dynamic programming. In fact, distributionally robust optimization (DRO) with ambiguity sets defined via MMD and kernel mean embedding has only been recently studied \cite{ZJDS21}, where the authors established the strong duality result for this class of problems. Notably, kernel DRO problems has not received much attention, with \cite{CKA22} being the sole exception.}

\changesroundOne{Before summarizing our main contributions, we highlight the widespread use of kernel methods, MMD, and kernel mean embedding within the machine learning and system identification literature \cite{Pillonetto2011,SGSS07,GLB12,Muandet2017}. For instance, Muandet et al. \cite{Muandet2017} express the expectation of any function of the subsequent state as an inner product between such a function and the conditional mean embedding. When the transition probability is unknown, and we have access to state-input trajectories, the empirical estimate of the conditional mean embedding has been used to approximate the inner product computation \cite{Thorpe2019,Thorpe2022}.} 

\changesroundOne{Moreover, distributionally robust dynamic programming has been studied in several recent contributions \cite{Yang2020,Yang2018,Schuurmans2023,FL22}, where ambiguity sets are defined in an exogenous manner, independent of the current state and action. While this is a reasonable assumption when the state evolution is uncertain in a parametric manner (e.g., the state transition being governed by known dynamics affected by a parametric uncertainty or additive disturbance), the more general case of state evolution given by a stochastic transition kernel requires defining the stochastic state evolution and its associated ambiguity set as a function of the current state and chosen action. This class of ambiguity sets is referred to as {\it decision-dependent ambiguity sets} and often poses challenges for tractable reformulations \cite{Luo2020,Noyan2022}.}
\changesroundOne{Our main contributions are summarized as: (1) We introduce a novel framework that provides a  solution to discrete-time stochastic control problems based on available system trajectories. The key objective is to derive robust control policies against sample errors due to finite data, achieved through the construction of ambiguity sets; (2) We simultaneously employ kernel mean embedding and MMD distance in space of probability distributions, enabling the formulation and solution to a distributionally robust dynamic programming; (3) We leverage certain functional analytic tools and an existing result in the literature to demonstrate that the set of optimal policies obtained through our formulation can be chosen to be Markovian and deterministic.}

\section{Preliminaries} 
\label{sec:prelim}

\subsection{Reproducing kernel Hilbert spaces (RKHS) and kernel mean embeddings}


Let $(\mathcal{X},\mathcal{F}_X)$ be a measurable space, where $\mathcal{X}$ is an abstract set and $\mathcal{F}_X$ represents a $\sigma$-algebra on $\XX$ \changesroundOne{(please refer to Chapter 1 in \cite{Salamon16} for more details)}. Consider a measurable function $k: \mathcal{X} \times \mathcal{X} \mapsto \mathbb{R}$, called a {\it kernel}, that satisfies the following properties.
\begin{itemize}
	\item \textit{Boundedness}: For any $x\in \XX$, we have that $\sup_x |k(x,x)| < \infty$.
    \item \textit{Symmetry}: For any $x, x' \in \mathcal{X}$, we have $k(x,x') = k(x',x)$;
    \item \textit{Positive semidefinite (PSD)}: For any finite collection of points $(x_i)_{i = 1}^m$, where $x_i \in \mathcal{X}$ for all $i = 1, \ldots, m$, and for any vector $\alpha \in \mathbb{R}^m$, we have that 
    \[
    \sum_{i,j = 1}^m \alpha_i \alpha_j k(x_i,x_j) \geq 0.
    \]
    In other words, the Gram matrix $K$ whose $(i,j)$-th entry is given by $k(x_i,x_j)$ is a positive semidefinite matrix for any choice of points $\{x_i\}_{i=1}^m$.
\end{itemize}
A kernel function $k$ satisfying the above three properties is called a {\it positive \changesroundOne{semi}definite} kernel. 


Two consequences are in place with the presence of a positive semidefinite kernel. First, every semipositive definite kernel is associated with a {\it reproducing kernel Hilbert space} (RKHS) 
$\HH_{\XX}$, which is defined as


\changesroundOne{\begin{equation}
	\HH_{\XX} = \overline{\bigcup_{\substack{I \subset \XX \\ I \text{ finite}}} \mathrm{span}\{ k(x,\cdot):\XX \mapsto \mathbb{R},  x \in I \}},
	\label{eq:RKHS-def}
\end{equation}}

\noindent defined as the closure of all possible finite dimensional subspaces induced by the kernel $k$. The space $\HH_{\XX}$ is equipped with the inner product $\innerRKHS{k(\cdot,x_1)}{k(\cdot,x_2)} = k(x_1,x_2)$. We may notice that the boundedness, symmetry and PSD properties above ensure that such an inner product is well-defined; hence, $\HH_{\XX}$ equipped with this inner product is a Hilbert space (see \cite{Brezis11} for more details). By definition, for any function $f \in \HH_{\XX}$ \changesroundOne{there exist a sequence of integers $\{m_n\}_{n \in \mathbb{N}}$ and a sequence of functions}  $f_n(x) = \sum_{i = 1}^{m_n} \beta_i^n k(x_i^n,x)$ such that $f(x) = \lim_{n \to \infty} f_n(x)$, which then implies that 
\[
\innerRKHS{f}{k(\cdot,x)}  = f(x),
\]
thus justifying the reproducing kernel property of the Hilbert space $\HH_{\XX}$. Alternatively, \changesroundOne{by the Riesz representation theorem (Theorem 4.11 in \cite{Brezis11}), any Hilbert space with the reproducing kernel property can be identified with the RKHS of a PSD kernel.} The interested reader is referred to  \cite{SGSS07} and references therein for more details. Second, there exists a feature map $\phi: \XX \mapsto \HH_{\XX}$ with the property $k(x_1,x_2) = \innerRKHS{\phi(x_1)}{\phi(x_2)}$. The canonical feature map is given by $\phi(x) := k(x,\cdot)$\changesroundOne{, and we use the notation $\phi(x)(x') = k(x,x')$.}


We now introduce the notion of {\it kernel mean embedding} of probability measures \cite{SGSS07,Muandet2017}. Let $\mathcal{P}(\mathcal{X})$ be the set of probability measures on $\mathcal{X}$, and $X$ be a random variable defined on $\XX$ with distribution $\mathbb{P}$. The kernel mean embedding is a mapping $\Psi: \mathcal{P}(\mathcal{X}) \mapsto \mathcal{H}_{\XX}$ defined as
\begin{equation}
    \Psi(\mathbb{P})(\cdot) := \mathbb{E}_{\mathbb{P}}[\phi(X)] = \int_{\mathcal{X}} k(x,\cdot) d \mathbb{P}(x).
    \label{eq:KernelMeanEmbedding}
\end{equation}

We have the following result from \cite{SGSS07,Muandet2017} on the reproducing property of the expectation operator in the RKHS. 

\begin{lemma}[Lemma 3.1 \cite{Muandet2017}]
If $\mathbb{E}_{\mathbb{P}}[\sqrt{k(X,X)}] < \infty$, then $\Psi(\mathbb{P}) \in \HH_{\XX}$ and $\mathbb{E}_{\mathbb{P}}[f(X)] = \innerRKHS{f}{\Psi(\mathbb{P})}$. 
\label{lemma:aux-lemma-KME}
\end{lemma}

\changesroundOne{Lemma \ref{lemma:aux-lemma-KME} implies that} $\Psi(\mathbb{P})$ is indeed an element of the RKHS $\HH_{\XX}$, and that the expectation of any function of the random variable $X$ can be computed by means of an inner product between the corresponding function and the kernel mean embedding. \changesroundOne{Let $\{ \widehat{x}_{(1)}, \ldots, \widehat{x}_{(m)}\}$ be a collection of $m$ independent samples from the distribution $\mathbb{P}$ and the consider the random variable $\widehat{\Psi}: \XX^m \times \mathcal{P}(\XX) \mapsto \HH_{\XX}$ as}

\begin{equation}
    \widehat{\Psi}(\mathbb{P})(\cdot) = \frac{1}{m} \sum_{i = 1}^m \phi(\widehat{x}_{(i)}) = \frac{1}{m} \sum_{i = 1}^m k(\widehat{x}_{(i)},\cdot).
    \label{eq:empirical_KME}
\end{equation}
In other words, $\widehat{\Psi}(\mathbb{P})$ is the mean embedding of the empirical distribution $\Pbhat_m := \frac{1}{m}\sum^m_{i=1} \delta_{\widehat{x}_{(i)}}$ induced by the samples.


\subsection{Kernel-based ambiguity sets}

\changesroundOne{There are several ways to introduce a metric in the space of probability measures (see \cite{Panchenko19} for more details). In the context of RKHS and the kernel mean embedding introduced in the previous section, a common metric is a type of integral probability metric \cite{SGSS07} called the \emph{Maximum mean discrepancy (MMD)}, which is defined as}
\changesroundOne{\begin{align}
    \mathtt{MMD}(\Pb,\Qb) &= \sup_{\|f\|_{\HH_{\XX}} \leq 1} \innerRKHS{f}{\Psi(\Pb)} - \innerRKHS{f}{\Psi(\Qb)} \nonumber \\  &= ||\Psi(\Pb) - \Psi(\Qb)||_{\HH_{\XX}}. 
    \label{eq:MMD-def}
\end{align}
\noindent In this work}, we consider data-driven MMD ambiguity sets induced by observed samples \changesroundOne{$\{\widehat{x}_{(1)},\ldots, \widehat{x}_{(m)}\}$} defined as
\changesroundOne{\begin{equation}\label{eq:mmd-ambiguity-set}
\mathcal{M}^{\changesroundOne{\epsilon}}_m := \{\Pb \in \mathcal{P}(\XX)~|~\mathtt{MMD}(\Pb,\Pbhat_m) \leq \epsilon\},
\end{equation}
\noindent where $\Pbhat_m$ is the empirical distribution defined earlier}. Thus, $\MM^\epsilon_m$ contains all distributions whose kernel mean embedding is within distance $\epsilon \geq 0$ of the kernel mean embedding of the empirical distribution. The above ambiguity set also enjoys a \changesroundOne{sharp uniform convergence guarantees} of $\mathcal{O}\left(\frac{1}{\sqrt{m}}\right)$ as shown in \cite{Nemmour2022}. Details are omitted for brevity.

\subsection{RKHS embedding of conditional distributions}


We now consider random variables of the form $(Y,X)$ taking values over the space $\YY \times \XX$. Let $\HH_{\YY}$ be the RKHS of real valued functions defined on $\YY$ with positive definite kernel $k_\YY: \YY \times \YY \mapsto \Rb$ and feature map $\phi_{\YY}: \YY \mapsto \HH_{\YY}$. \changesroundOne{For any distribution $\mu \in \mathcal{P}(\YY \times \XX)$, we define the bilinear covariance operator \cite{Baker1973} $\mathrm{cov}: \HH_\YY \times \HH_\XX \mapsto \mathbb{R}$ as}
\begin{equation}
    \mathrm{cov}(g,f) = \int_{\YY \times \XX} g(y) f(x) \mu(dy,dx).
    \label{eq:cov-operator-def}
\end{equation}

\changesroundOne{For each fixed $g \in \HH_{\YY}$, we define a cross-covariance operator $C_{XY}: \HH_{\YY} \mapsto \HH_{\XX}$ as the unique operator, which is well-defined due to Riesz representation theorem, such that, for all $f \in \HH_{\XX}$, we have}
\changesroundOne{\[\innerRKHS{f}{C_{XY}g} = \mathrm{cov}(g,f).\]
\noindent Similarly, we define the operator $C_{YX}: \HH_{\XX} \mapsto \HH_{\YY}$, and one may notice that $C_{YX} = C_{XY}^\star$, where $C_{XY}^\star$ is the adjoint operator of $C_{XY}$.} When $\YY=\XX$, the analogous operator $C_{YY}$ is called the covariance operator. \changesroundOne{Of interest to our approach is} the conditional mean embedding of the conditional distributions.

\changesroundOne{For two measurable spaces $(\XX,\mathcal{B}_\XX)$ and $(\YY,\mathcal{B}_{\YY})$, where $\mathcal{B}_{\XX}$ and $\mathcal{B}_{\YY}$ are associated $\sigma$-algebras, a stochastic kernel $T:\YY \rightarrow \mathcal{P}(\XX)$ is a measurable mapping from $\YY$ to the set of probability measures equipped with the topology of weak convergence. Please refer to \cite{Bertsekas.Shreve78}, Chapter 7, for a more detailed definition.} 




\begin{definition}[Definition 4.1 \cite{Muandet2017}]
\changesroundOne{Given a stochastic kernel $T:\YY \mapsto \mathcal{P}(\XX)$, its conditional mean embedding is a mapping $\CME{}: \YY \mapsto \HH_{\XX}$ such that, for all $f \in \HH_{\XX}$, the following property holds
\begin{equation}
    \innerRKHS{\psi(y)}{f} = \int_{\XX} f(x) T(dx \mid y) = \Eb_{X \sim T(\cdot \mid y)}[f(X)],
    \label{eq:CME-def}
\end{equation}    
that is, the inner product of the conditional mean embedding with a function in $f \in \HH_{\XX}$ coincides with the conditional expectation of $f$ under the stochastic kernel $T$.}
\label{defi:CME}
\end{definition}

\changesroundOne{A crucial result to the success of conditional mean embedding in machine learning applications is the relation proved in \cite{Baker1973,Fukumizu2004} between the mappings in Definition \ref{defi:CME} and the covariance and cross-covariance operators $C_{XY}$ and $C_{YY}$ as shown in the sequel.
\begin{proposition}[Corollary 3, \cite{Fukumizu2004}]
    Let $\Qb \in \mathcal{P}(\YY)$ be a probability measure over $\YY$ and $T:\YY \mapsto \mathcal{P}(\XX)$ be a stochastic kernel. Define the measure $\mu \in \mathcal{P}(\YY \times \XX)$ as the unique measure such that $\mu(U \times V) = \int_{U} \Qb(dy) T(V \mid y)$ and consider its corresponding covariance operators as in \eqref{eq:cov-operator-def}. Then we have that 
        \[
            \psi(y)(\cdot) = C_{XY} C_{YY}^{-1} k_{\YY}(\cdot,y),
        \]
    where $\psi$ is the conditional mean embedding of Definition \ref{defi:CME}.
    \label{prop:CME-main-result}
\end{proposition}}


\changesroundOne{The importance of Proposition \ref{prop:CME-main-result} lies in the fact that it allows us to estimate the mapping $\psi$ by means of available data. In fact,} in many applications, the joint or conditional distributions involving $\YY$ and $\XX$ are not known, rather we have access to i.i.d. samples \changesroundOne{$\{(\widehat{y}_{(i)},\widehat{x}_{(i)})\}_{i = 1}^m$ drawn from the joint distribution $\mu \in \mathcal{P}(\YY \times \XX)$ induced by an underlying stochastic kernel (namely, the system dynamics, as discussed in the next section)}. 
\changesroundOne{ \begin{remark}
    The main step of the proof of Proposition \ref{prop:CME-main-result} consists in showing that, for all $f \in \HH_{\XX}$, the relation $\innerRKHS{f}{C_{XY}C_{YY}^{-1}k_{\YY}(\cdot,y)} = \innerRKHS{f}{\psi(y)}$ holds. Details are omitted for brevity. 
     \label{rmk:CME}
 \end{remark}}
Let $K_{\changesroundOne{\YY}} \in \Rb^{\changesroundOne{m}\times \changesroundOne{m}}$ be the Gram matrix associated with $\{\widehat{y}_i\}_{i = 1}^{\changesroundOne{m}}$, that is, the matrix whose $(i,j)$-th entry given by \changesroundOne{$[K_\YY]_{ij} = k_\YY(\widehat{y}_{(i)},\widehat{y}_{(j)})$}. The empirical estimate of the conditional mean embedding $\psi(x)$ was derived in \cite{SHSF09} and is stated below.

\begin{theorem}[Theorem 4.2 \cite{Muandet2017}]
\label{theo:conditional-mean-embedding}
An empirical estimate of the conditional mean embedding $\psi : \YY \rightarrow \HH_{\XX}$ is given by
\begin{equation}\label{eq:def_empirical_c_kme}
    \changesroundOne{\widehat{\psi}_m(y)}(\cdot) = \sum^m_{i=1} \beta_i(y) k_\XX(\widehat{x}_{(i)},\cdot),
\end{equation}
where $\beta(y) = (K_\YY+m \lambda \mathbf{I}_m)^{-1} k_{\mathbf{y}} (y) \in \Rb^m$ with \changesroundOne{$k_{\mathbf{y}}(y) = [k_\YY(\widehat{y}_{(1)},y) \quad k_\YY(\widehat{y}_{(2)},y) \quad \ldots k_\YY(\widehat{y}_{(m)},y)]^\top \in \Rb^m$}, $\mathbf{I}_m$ being the identity matrix of dimension \changesroundOne{$m$} and $\lambda > 0$ being a regularization parameter. 
\end{theorem}

The above empirical estimate can also be obtained by solving a regularized regression problem as established in  \cite{GLB12b,Micchelli2005}. \changesroundOne{The regularization is necessary for the well-posedness of the empirical estimator of the inverse of the operator $C_{YY}$.}


\section{Kernel Distributionally Robust Optimal Control}
\label{sec:problem-formulation}

\changesroundOne{We now connect the abstract mathematical framework introduced in the previous sections in the context of optimal control problems. Let $(\XX, \mathcal{F}, \mathbb{P})$ be a probability measure space, where $\XX \subset \Rb^\dimState$; we use such measure-theoretic framework to formalize statements about the stochastic process in the sequel.
Consider the discrete-time stochastic system given by} 
\begin{equation}
    x_{k+1} \sim T(\cdot|x_k,a_k), \quad a_k \in \mathcal{A}(x_k), \quad x_0 = \bar{x},
    \label{eq:DT_stochasticSystem}    
\end{equation}
where $x_k \in \XX \subset \Rb^\dimState$ and $a_k \in \mathcal{A}(x) \subset \Rb^\dimControl$ denote the state and control input at time $k \in \mathbb{N}$, with $\mathcal{A}(x_k)$ being the set of admissible control inputs at time $k$, and $\bar{x}$ denotes the initial state. \changesroundOne{Observe that we denote the system dynamics by a stochastic kernel $T: \XX \times \mathcal{A}(\XX) \rightarrow \mathcal{P}(\XX)$ defined over the state-input pair, which describes a probability distribution over the next state. It is a standard machinery to assign the semantics for the dynamical system in \eqref{eq:DT_stochasticSystem} in such a way that it defines a unique probability measure in the space of sequences with value in $\XX$ using the Kolmogorov extension theorem. See \cite{Durrett10} for more details.}

\changesroundOne{We define a sequence of \textit{history-dependant (or non-Markovian)} control policies $\pi_k:(\XX\times \mathcal{A}(\XX))^{k} \times \XX \rightarrow \mathcal{P}(\mathcal{A}(\XX))$, which maps a sequence of state-input pair $(x_0,a_0,\ldots,x_{k-1},a_{k-1},x_k)$ into a probability measure with support in the space of admissible inputs $\mathcal{A}(\XX)$.} The collection of admissible control policies over a horizon of length $L$ is given by the set 
\changesroundOne{\begin{align*}
    \Pi_L = \{(\pi_0,\pi_1,&\ldots,\pi_{L-1})\mid\pi_k(\mathcal{A}(x_k)|h) = 1, \\ &\text{for all } k \in \{0,\ldots, L-1\}, \text{for all } \\ & \hspace{2cm}h \in (\XX \times \mathcal{A}(\XX))^{k-1}\times \XX\}, 
\end{align*}
composed by the set $L$ control policies with support on the admissible input set associated with the current system state. In case the control policy depends only on the current state, i.e., $\pi_k(h) = \pi_k(h')$ for all $h,h' \in (\XX \times \mathcal{A}(\XX))^{k-1}\times \XX$ that agree on the last entry, then we call such a policy a \textit{Markovian policy}.}

\changesroundOne{Inspired by a growing interest in data-driven control methods \cite{Thorpe2019,Thorpe2022,Persis2020}, our objective is to design an admissible control policy that minimizes a performance index with respect to the generated trajectories of the system dynamics given in \eqref{eq:DT_stochasticSystem} by relying on available data set $\{\widehat{x}_{(i)},\widehat{a}_{(i)},\widehat{x}_{(i)}^+\}_{i = 1}^m$, composed by a sequence of state-input-next-state tuple. To account for the sampling error due to the finite dataset, we leverage techniques from distributionally robust optimization \cite{HCL19,MohajerinEsfahani2018,Shapiro2017,Yang2018} by creating a  \textit{state-dependant ambiguity set} around the estimate of the system dynamics, which is obtained using the mathematical framework introduced in Section \ref{sec:prelim}.}





\changesroundOne{Formally, and referring to the notation used in Section \ref{sec:prelim}, let us identify the space $\YY$ with the cartesian product $\XX \times \mathcal{A}(\XX)$ and $\XX$ with the state space of the dynamics in \eqref{eq:DT_stochasticSystem}. Let $k_\YY: (\XX \times \mathcal{A}(\XX)) \times (\XX \times \mathcal{A}(\XX)) \rightarrow \mathbb{R}$ and $k_{\XX}: \XX \times \XX \rightarrow \mathbb{R}$ be the corresponding kernels that induce, respectively, the RKHS $\HH_\YY$ and $\HH_\XX$. Let $T: \XX \times \mathcal{A}(\XX) \rightarrow \mathcal{P}(\XX)$ be the stochastic kernel associated with the state-transition matrix, and $\psi$ and $\widehat{\psi}$ be the conditional mean embeddings as in Definition \ref{defi:CME} and Theorem \ref{theo:conditional-mean-embedding}, respectively; the latter obtained from the dataset $\{\widehat{x}_{(i)},\widehat{a}_{(i)},\widehat{x}_{(i)}^+\}_{i = 1}^{m}$.} 

\changesroundOne{A key novelty of our approach is the construction of an ambiguity set using the kernel mean embedding introduced in \eqref{eq:KernelMeanEmbedding}: 
\begin{equation}
    \MM^\epsilon(x,a) =  \{ \Pb \in \mathcal{P}(\XX) \mid \|\Psi(\Pb) - \psi(x,a) \|_{\HH_\XX} \leq \epsilon \},
    \label{eq:ambiguity-set-def}
\end{equation}
which is a collection of probability measures over $\XX$ whose mean embedding is $\epsilon$ close to the conditional mean embedding of the system dynamics. Notice the dependence of the current state-action pair through the conditional mean embedding. Similar to the control policy definition, we define a collection of \textit{admissible dynamics} for a given time-horizon $L$. Formally, for a given sequence $(x_0,a_0,\ldots,x_{L-1},a_{L-1})$ of state-input pairs of size $L$ we define the set of admissible dynamics as 
}
\begin{align}
    \Gamma_L = \left\{(\mu_0,\mu_1,\ldots,\mu_{L-1}) \mid \mu_k \in \MM^\epsilon(x_k,a_k) \right. \nonumber\\
    \left. \text{for all } k \in \{0,\ldots,L-1\} \right\},
    \label{eq:admissible-dynamics}
\end{align}
where $\mathcal{M}^\epsilon(x_k,a_k)$ denotes the \changesroundOne{ambiguity set in \eqref{eq:ambiguity-set-def}.}

\changesroundOne{Our main goal is to solve a distributionally robust dynamic programming using the ingredients presented so far. To this end, let $L \in \mathbb{N}$ be the time-horizon, and consider the corresponding set of control policies $\Pi_L$ and admissible dynamics $\Gamma_L$. For any $\pi \in \Pi_L$ and $\mu \in \Gamma_L$, and for any initial state $\bar{x} \sim \mu_0$, where $\mu_0$ is the first entry of the admissible dynamics $\mu$, we denote by $\Pb^{\pi,\mu}$ the induced measure on the space of sequences in $\XX$ of length $L$.} For a given stage cost $c: \changesroundOne{\XX \times \mathcal{A}(\XX)} \to \Rb$, we define the finite horizon expected cost as
\begin{equation}
    V_{\changesroundOne{L}}\changesroundOne{(\pi,\mu)} := \Eb^{\changesroundOne{\pi,\mu}}\left[\sum^{L-1}_{k=0} c(x_k,a_k) \right],
    \label{eq:value-function-def}
\end{equation}
where \changesroundOne{$\Eb^{\pi,\mu}$} denotes the expectation operator with respect to \changesroundOne{$\Pb^{\pi,\mu}$}. Our goal is to find a policy $\pi^\star \in \changesroundOne{\Pi_L}$ that solves the distributionally robust control problem given by
\changesroundOne{\begin{equation}\label{eq:minmax_control}
    \inf_{\pi \in \Pi_L} \sup_{\mu \in \Gamma_L} V_L(\pi,\mu). 
\end{equation}}
\changesroundOne{Problem \eqref{eq:minmax_control} consists of a min-max infinite-dimensional problem, since the collection of control policy and admissible dynamics are infinite-dimensional spaces; and it has been studied in several related communities \cite{Shapiro2017} under different lenses. This paper departs from related contributions within the control community \cite{Thorpe2022,Thorpe2019}, showing that under technical conditions and leveraging compactness-related arguments that there exist optimal Markovian policies for the solution of \eqref{eq:minmax_control}.} \changesroundOne{To this end, we consider the following technical conditions. } 

\begin{assumption}
    \label{assumption:main}
    Let $\changesroundOne{V \subset }~\XX \times \mathcal{A}(\XX) \times \mathcal{P}(\XX)$ \changesroundOne{be a weakly compact set such that $(x,a,\mu) \in V$ if and only if $\mu \in \MM^\epsilon(x,a)$ for some $\epsilon > 0$}. The following conditions hold.
        \begin{enumerate}
            \item The stage cost function $c(x,a)$ \changesroundOne{in \eqref{eq:value-function-def}}  is lower semicontinuous. Besides, there exists a constant $C \geq 0$ and a continuous function \changesroundOne{$w:\XX \rightarrow [1,\infty)$} such that 
                \[ |c(x,a)| \leq C w(x), \quad \forall a \in \mathcal{A}(x),~x \in \XX.\]

            \item  For every bounded, continuous function \changesroundOne{$f: \XX \to \Rb$, the functional $\varphi: V \rightarrow \Rb$, defined as $(x,a,\mu)  \mapsto  \int_{\XX} f(\xi) \mu(d\xi)$ is continuous on $V$.} 
            \item There exists \changesroundOne{a constant $C > 0$ such that $\varphi(x,a,\mu) \leq C w(x)$ for all $(x,a,\mu) \in V$.} 
            \item The set $\mathcal{A}(x)$ is compact for each $x \in \XX$. Furthermore, the set-valued mapping $x \mapsto \mathcal{A}(x)$ is upper semi-continuous.
        \end{enumerate}
\end{assumption}

\begin{theorem}
    Suppose Assumption \ref{assumption:main} holds. Then there exists a collection of functions $f_k:\XX \rightarrow \mathbb{R}^{\dimControl}$, for $k \in \{0,\ldots, L-1\}$, such that the Markov policy $(f_0,\ldots,f_{L-1})$ is the optimal solution of the infinite-dimensional problem \eqref{eq:minmax_control}.
    \label{theo:main-result}
\end{theorem}
\begin{proof}
    
    Our proof strategy is inspired by the analysis in \cite{GonzalezTrejo2002}. Without loss of generality, we may assume that the cost function $c: \XX \times \mathcal{A} \rightarrow \mathbb{R}^n$ is bounded, that is, $w \equiv 1$, (otherwise one could consider a weighted norm using the function $w$ appearing in item $a$ of Assumption \ref{assumption:main} -- see \cite{GonzalezTrejo2002} for more details.) For each bounded function $g: \XX \rightarrow \mathbb{R}$, we define the functionals $H: \HH_\XX \times \XX \times \mathcal{A} \times \mathcal{P}(\XX) \rightarrow \mathbb{R}$ and $ H^{\#}: \HH_\XX \times \XX \times \mathcal{A} \rightarrow \mathbb{R}$, and the mapping $H^{\dag}:\HH_\XX \rightarrow \HH_\XX$ as 
    \begin{align}
        & H(g,x,a,\nu) := c(x,a) + \int_{\XX} g(\xi)\nu(d \xi), 
        \\ & H^{\#}(g,x,a) := \sup_{\nu \in \mathcal{M}^\epsilon(x,a)} H(g,x,a,\nu),
        \\ & H^{\dag}(g)(x) := \inf_{a \in \mathcal{A}(x)} H^{\#}(g,x,a) \nonumber
        \\ & \quad = \inf_{a \in \mathcal{A}(x)} \sup_{\nu \in \mathcal{M}^\epsilon(x,a)} \left[ c(x,a) + \int_{\XX} g(\xi) \nu(d\xi) \right]. \label{eq:mmd_dp_def}
    \end{align}
    Specifically, \eqref{eq:mmd_dp_def} defines the distributionally robust dynamic programming (DP) operator under MMD ambiguity defined in \eqref{eq:ambiguity-set-def}, which is an infinite-dimensional optimization problem. Given any initial function $f \in \HH_\XX$, the value function of the distributionally robust control problem can be defined iteratively as
    \begin{align}
        v_{L-1}(x) & := f(x), \text{ for all } x \in \XX, \nonumber \\
        v_k(x) & := H^\dag(v_{k+1})(x), \text{for all } x \in \XX, \label{eq:DP_def}
    \end{align}
    for $k = 0, \ldots, L-2$. We are now ready to show that problem \eqref{eq:minmax_control} admits a non-randomized, Markov policy which is optimal. To this end, we need to show (1) lower semi-continuity of the functional $H$ and (2) that the inner supremum in the definition of $H^\dag$ is achieved.

    To show lower semi-continuity, we follow a similar approach as the proof of \cite[Theorem 3.1]{GonzalezTrejo2002} and \cite[Theorem 1]{Yang2018}. The primary challenge is to show that the functional $H^\dag$ used in \eqref{eq:DP_def} preserves the lower semi-continuity of the value function. We show this via induction. Let $v_{k+1}$ be lower semicontinuous. Following identical arguments as \cite[Lemma 3.3]{GonzalezTrejo2002}, it can be shown that $H^\dag(v_{k+1})$ is lower semicontinuous on $\XX$, provided one can show that the mapping $(x,a) \mapsto \mathcal{M}^\epsilon(x,a)$ is weakly\footnote{Due to the fact that we are dealing with some infinite-dimensional spaces, we rely on the topological notion of continuity. A function between two topological spaces (please refer to \cite{Munkres74} for an introduction to these concepts) $(\YY,\tau_\YY)$ and $(\XX,\tau_\XX)$, $f: \YY \mapsto \XX $ is continuous if for all $U \in \tau_{\XX}$ we have that $f^{-1}(U) \in \XX$. The notions of weakly continuous, weakly compact, etc, are used due to the fact that we equip the infinite-dimensional spaces $\mathcal{P}(\XX)$ and $\HH_{\XX}$ with the $\text{weak}^*$ topology. We refer the reader to \cite{Brezis11}, Chapter 4, for more details about these concepts.} compact and lower semi-continuous (i.e., the condition analogous to \cite[Assumption 3.1(g)]{GonzalezTrejo2002}).

    To show weak compactness of $(x,a) \mapsto \mathcal{M}^\epsilon(x,a)$, let us first study properties of the set 
    \[
    \mathcal{C}^\epsilon(x,a) = \{ f \in \mathcal{H}_{\mathcal{X}}: \| f - \psi(x,a) \|_{\HH_\XX} \leq \epsilon \},
    \]
    where $\psi$ is the conditional mean embedding mapping of Definition \ref{defi:CME} using the notation of Section \ref{sec:problem-formulation}. Note that $\mathcal{C}^\epsilon(x,a)$ is a convex subset of the RKHS $\mathcal{H}_{\mathcal{X}}$; hence, by \cite[Theorem 3.7]{Brezis11}, it is also weakly closed. Since Hilbert spaces are  reflexive Banach spaces, then by Kakutani's theorem (see \cite[Theorem 3.17]{Brezis11}) we show that the set $\mathcal{C}^\epsilon(x,a)$ is weakly compact. Now, observe that 
    \[
    \mathcal{M}^\epsilon(x,a) = \Psi^{-1}(\mathcal{C}^\epsilon(x,a)),
    \]
     where the mapping $\Psi: \mathcal{P}(\mathcal{X}) \mapsto \mathcal{H}_{\mathcal{X}}$, defined in \eqref{eq:KernelMeanEmbedding}, is continuous due to boundedness and continuity of the kernel, thus weakly continuous. Since $\Psi$ is also surjective\footnote{For any function $f \in \mathcal{H}_{\mathcal{X}}$ there exists a sequence $f_n(x) = \sum_{i = 1}^m \beta_i(n)k(x_i,x)$ such that $\lim_{n \to \infty}f_n(x) = f(x)$. Let $\mathbb{P}_n = \sum_{i = 1}^m \beta_i(n) \delta_{x_i}$ and notice that $\Psi(\mathbb{P}) = \Psi(\lim_{n \to \infty} \mathbb{P}_n) = f$, thus showing that $\Psi$ is a surjective mapping.}, we have by the open mapping theorem (\cite[Theorem 2.6]{Brezis11}) that $\mathcal{M}^\epsilon(x,a)$ is also weakly compact.

    For a fixed $\epsilon >0$, we now show that the mapping $(x,a) \mapsto \mathcal{M}^\epsilon(x,a)$ is lower semicontinuous. We define the distance function from a distribution $\nu \in \mathcal{P}(\XX)$ to a convex subset $S$ of $\mathcal{M}^\epsilon(x,a)$ as\footnote{This distance is only well-defined due to the weak compactness result shown above, as it would allow us to take convergent subsequences for any sequence achieving the infimum in this definition.} 
    $$ d(\nu,S) := \inf_{\mu \in S} ||\Psi(\mu) - \Psi(\nu)||_{\HH_{\XX}}. $$ 
    Let $(x,a,\nu)$ be given, with $x\in\XX$, $a\in \mathcal{A}(x)$, and $\nu \in \mathcal{M}^\epsilon(x,a)$. Thus, 
    $$ ||\Psi(\nu) - \psi(x,a)||_{\HH_{\XX}} \leq \epsilon. $$ 

    Consider a sequence $(x_n,a_n)_{n \geq 0}$ with $a_n \in \mathcal{A}(x_n), \forall n \geq 0$ and $\lim_{n \to \infty} (x_n,a_n) = (x,a)$. From \cite[Proposition 1.4.7]{Aubin.Frankowska09}, it follows that the lower semi-continuity of $(x,a) \mapsto \mathcal{M}^\epsilon(x,a)$ is equivalent to
    $$ \nu \in \liminf_{n \to \infty} \mathcal{M}^\epsilon(x_n,a_n) \iff \lim_{n \to \infty} d(\nu,\mathcal{M}^\epsilon(x_n,a_n)) = 0.$$
    To this end, we compute
    \begin{align*}
        & ||\Psi(\nu) - \psi(x_n,a_n)||_{\HH_{\XX}} \leq ||\Psi(\nu) - \psi(x,a)||_{\HH_{\XX}} 
        \\ & \qquad \qquad + ||\psi(x,a) - \psi(x_n,a_n)||_{\HH_{\XX}}
        \\ \implies & \lim_{n \to \infty} ||\Psi(\nu) - \psi(x_n,a_n)||_{\HH_{\XX}} \leq \epsilon 
        \\ & \qquad \qquad + \lim_{n \to \infty} ||\psi(x,a) - \psi(x_n,a_n)||_{\HH_{\XX}}
        \\ \implies & \lim_{n \to \infty} ||\Psi(\nu) - \psi(x_n,a_n)||_{\HH_{\XX}} \leq \epsilon
        \\ \implies & \lim_{n \to \infty} d(\nu,\mathcal{M}^\epsilon(x_n,a_n)) = 0,
    \end{align*}
    from the definition of the ambiguity set. In the second last step, the second term goes to $0$ as $\lim_{n \to \infty} (x_n,a_n) = (x,a)$ due to the continuity of the kernel function and the definition of $\psi(x,a)$ in Proposition \ref{prop:CME-main-result}. As a result, $\nu \in \liminf_{n \to \infty} \mathcal{M}^\epsilon(x_n,a_n)$ and thus, $\mathcal{M}^\epsilon(x,a) \subseteq \liminf_{n \to \infty} \mathcal{M}^\epsilon(x_n,a_n)$. 

    With the above properties of the mapping $(x,a) \mapsto \mathcal{M}^\epsilon(x,a)$ in hand, it can be shown following identical arguments as the proof of \cite[Theorem 3.1]{GonzalezTrejo2002} that both $H^{\#}(v_{k+1},x,a)$ and $\mathcal{H}^\dag(v_{k+1})$ are lower semicontinuous and there exists a function $f_{k}: \XX \to \mathcal{A}(x)$ such that $a_{k} = f_{k}(x_{k})$ is the minimizer of \eqref{eq:mmd_dp_def} for $u = v_{k+1}$. This concludes the proof.

\end{proof}

\begin{remark}
    Note that the ambiguity sets considered in related works on distributionally robust control \cite{Yang2018,Yang2020} do not depend on the current state and action, i.e., the mapping $(x,a) \mapsto \mathcal{M}^\epsilon(x,a)$ is independent of $(x,a)$, and as a result, properties such as compactness and lower semi-continuity are easily shown. In contrast, the ambiguity set considered here is more general and directly captures the dependence of the transition kernel on the current state-action pair. In the DRO literature, such ambiguity sets are referred to as {\it decision-dependent ambiguity sets} \cite{Wood2021,Noyan2022,Luo2020}, which are relatively challenging to handle, and less explored in the past. 
\end{remark}

\changesroundOne{By showing that there exist deterministic and Markovian policies in the optimal set of \eqref{eq:minmax_control}, Theorem \ref{theo:main-result} allows to reduce the search of general history-dependant policies to functions $\pi:\XX \rightarrow \mathcal{A}(\XX)$, which maps at each time instant a state to an admissible action. However, the result of Theorem \ref{theo:main-result} does not lead to a computationally tractable formulation towards solving problem \eqref{eq:minmax_control}. To this end, we now describe two approximation schemes that enable us to obtain a computational solution to \eqref{eq:minmax_control}.
The first approximation is related to the fact that we are aiming towards a data-driven solution of \eqref{eq:minmax_control}; hence, the ambiguity set in \eqref{eq:ambiguity-set-def} is replaced by its sample-based counterpart, where its center is defined according to the estimate of the conditional mean embedding mapping introduced in Proposition \ref{prop:CME-main-result}. That is, we consider instead the ambiguity set
\begin{equation}
    \widehat{\MM}^{\epsilon}_m(x,a) = \{\Pb \in \mathcal{P}(\XX)\mid \| \Psi(\Pb) - \widehat{\psi}_m(x,a) \|_{\HH_{\XX}} \leq \epsilon\}.
    \label{eq:ambiguity-set-empirical-def}
\end{equation}}

\changesroundOne{Under this approximation and exploiting the rectangular structure of the admissible dynamics in \eqref{eq:admissible-dynamics}, one may observe that the inner supremum in \eqref{eq:ambiguity-set-def} is computed separately for each time instant and that if we replace $\MM^\epsilon$ by  $\widehat{\MM}^{\epsilon}_m$ in the corresponding inner optimization we obtain, for any $f' \in \HH_\XX$,} 
\begin{align}
\sup_{\mu_k \in \widehat{\MM}^\epsilon_m(x_k,a_k)} \Eb_{X \sim \mu_k}[f'(X)] =  \sup_{\Psi(\mu_k) \in \widehat{\mathcal{C}}_k} \innerRKHS{f'}{\Psi(\mu_k)},
\label{eq:apprx-kernel-ambiguity}
\end{align}
\changesroundOne{\noindent where the last equality is a consequence of the reproducing property and} 
\[ \widehat{\mathcal{C}}_k := \{f \in \HH_{\XX} \mid ||f - \widehat{\psi}_m(x_k,a_k)||_{\HH_{\XX}} \leq \epsilon\} \subseteq \HH_{\XX}.\] 

Following \cite{ZJDS21}, the support of $\widehat{\mathcal{C}}_k$ can then be computed as
\begin{align*}
& \sigma_{\widehat{\mathcal{C}}_k}(f') = \sup_{f \in \widehat{\mathcal{C}}_k} \innerRKHS{f'}{f}
\\ & = \sup_{f \in \widehat{\mathcal{C}}_k} \quad \innerRKHS{f'}{f - \widehat{\psi}_m(x_k,a_k)} + \innerRKHS{f'}{\widehat{\psi}_m(x_k,a_k)}
\\ & = \epsilon ||f'||_{\HH_{\XX}} + \sum_{i = 1}^m \beta_i(x_k,a_k) f'(\changesroundOne{\widehat{x}_{(i)}^+}),
\end{align*}
where $\beta_i(x_k,a_k)$ denote the coefficients of the empirical estimate of the conditional mean embedding as given in \eqref{eq:def_empirical_c_kme} \changesroundOne{evaluated at the state-action pair at the time instant $k \in \{0, \ldots, L\}$, and where $\widehat{x}_{(i)}^+$ represents the collected sample\footnote{\changesroundOne{That is, for each $i\in \{1,\ldots,m\}$, $\widehat{x}_{(i)}^+$ is a sample from $T(\cdot \mid \widehat{x}_{(i)}, \widehat{a}_{(i)})$. Within the machine learning literature, the transition kernel $T$ is usually referred to as a ``generative model''.}} along the observed trajectories of the dynamics}. To compute the norm $\|f'\|_{\mathcal{H}_{\mathcal{X}}}$, we solve, similar to Theorem \ref{theo:conditional-mean-embedding}, a regression problem given by
\begin{align}
    \minimize_{\alpha_i, i = 1, \ldots, m} & \quad \left\| \sum_{i = 1}^m \alpha_i k(\widehat{x}_{(i)}',\cdot) - f'\right\|_{\mathcal{H}_{\mathcal{X}}} + \lambda \|\alpha\|_2^2,
    \label{eq:regression-RKHS-norm}
\end{align}
where \changesroundOne{$\{\widehat{x}_{(1)}', \ldots, \widehat{x}_{(m)}'\}$} are arbitrary points in the domain of the function $f'$. The solution of this regression problem is given by $\|f'\|_{\mathcal{H}_{\mathcal{X}}} = \sqrt{\alpha^\top K_\XX' \alpha}, $ where
\[
\quad \alpha = (K_\XX' + \lambda I)^{-1} \changesroundOne{\begin{bmatrix}
f'(\widehat{x}_{(1)}') &
\hdots &
f'(\widehat{x}_{(m)}')
\end{bmatrix}^\top}.
\]
Variable \changesroundOne{$K_\XX'$} is the 
\changesroundOne{$m \times m$} Gram matrix associated with the available points, namely, it is the positive semidefinite matrix whose $(i,j)$-entry is given by \changesroundOne{$k_\XX(\widehat{x}_{(i)}',\widehat{x}_{(j)}')$}. Notice that the collection of points used to estimate $\|f'\|_{\HH_\XX}$ may not necessarily coincide with those points used to estimate the conditional mean embedding in Theorem \ref{theo:conditional-mean-embedding}; \changesroundOne{we denote this in our notation by adding a prime as a superscript}. Same comment applies for the regularizer $\lambda$ appearing in \eqref{eq:regression-RKHS-norm} and in the expression of $\beta$ in Theorem \ref{theo:conditional-mean-embedding}.

We now discuss how to solve for optimal control inputs \changesroundOne{using a value iteration approach. To this end, we introduce our second approximation of the distributionally robust dynamic programming \eqref{eq:minmax_control} by discretizing the action space. At a given state $x \in \XX$, we discretize the input space $\mathcal{A}(x)$, denoting the resulting discrete set as $\{a^{(1)},\ldots,a^{(M)}\}$. To define the value function, we modify slightly the set of admissible dynamics $\Gamma_L$ in \eqref{eq:admissible-dynamics} by fixing the initial distribution to a given point $x \in \XX$. In other words, we let\footnote{ \changesroundOne{The symbol $\delta_x$, for some $x \in \XX$, represents the Diract measure, i.e., a probability measure that assigns unit mass to the point $x \in \XX$.}} $\mu_0 = \delta_{x}$ for some $x \in \XX$. Then, for a given function $f:\XX \rightarrow \Rb$, we define recursively the collection of functions $v_\ell: \XX \rightarrow \Rb$, where $\ell \in \{0, \ldots, L-2\}$, as $v_{L-1} = f$, and}
\begin{align}
    v_{\ell}(x) = \min_{a^{(j)}: j=1,\ldots,M}  &c(x,a^{(j)}) \nonumber \\ &+ \sup_{\mu \in \widehat{\MM}^{\epsilon}_m(x,a^{(j)})} \int_{\XX} \!\! v_{\ell+1}(\xi) \mu(d\xi)
\end{align}


For a given $(x,a^{(j)})$, the inner supremum problem is an instance of \eqref{eq:apprx-kernel-ambiguity} \changesroundOne{and can be computed using the techniques discussed in this section.}

\subsection{\changesroundOne{Distributionally robust approach to forward invariance or safe control}}

While the discussion thus far has focused on the general optimal control problem, this approach can also be leveraged for synthesizing control inputs meeting safety specifications. \changesroundOne{Using the notation of previous sections}, let $L \in \mathbb{N}$ be the time-horizon and $\safeSet \subset \mathbb{R}^\dimState$ be a measurable safe set. For an admissible control policy $\pi \in \Pi_L$, \changesroundOne{admissible dynamics $\mu \in \Gamma_L$,} and initial state $x$, the probability of the state trajectory being safe is given by (more details can be found in \cite{Abate2008})
\begin{align}
\changesroundOne{\safeProb_\safeSet(x;\pi,\mu)} & = \mathbb{P}^{\pi,\mu}_{x_0} \{ x_k \in \safeSet, \text{ for all } k \in \{0, \ldots, L-1\} \}, \label{eq:safeProb} 
\end{align}
where $(x_0,x_2,\ldots,x_{L-1})$ denote the solution of \eqref{eq:DT_stochasticSystem} \changesroundOne{under the dynamics $\mu$ and policy $\pi$. Our goal is to solve the problem} 
\changesroundOne{\begin{equation}
    \safeProbOpt_{\safeSet}(x) = \sup_{\pi \in \Pi_L} \inf_{\mu \in \Gamma_L} \safeProb_\safeSet(x;\pi,\mu),
    \label{eq:safeProbOpt}        
\end{equation}
using the mathematical framework proposed in this paper. In fact, one can show that an analogous result of Theorem \ref{theo:main-result} holds for \eqref{eq:safeProbOpt}, namely, there exists an optimal Markovian and deterministic policy. Hence, using similar approximations as in the previous section, function $V_{\safeSet}^\star$ in \eqref{eq:safeProbOpt} can be approximated recursively as (see \cite{Abate2008} for more details)
\begin{align}
    v_{L-1}&(x)  := \mathbf{1}_{\safeSet}(x),  \nonumber \\
    v_\ell(x) & := \max_{a^{(j)}: j = 1, \ldots, M} \inf_{\mu \in \widehat{\MM}^{\epsilon}_m (x,a^{(j)})} \mathbf{1}_{\safeSet}(x) \int_{\XX}v_{\ell + 1}(\xi) \mu(d\xi),  
    \label{eq:safety_DP}
\end{align}
for $\ell \in \{0,\ldots,L-2\}$, where $\mathbf{1}_{\safeSet}$ is the indicator function of the safe set $\safeSet$.}  In the numerical example reported in the following section, we compute the safe control inputs by solving the inner problem in an identical manner as discussed in the previous subsection.

\section{Numerical examples}
\label{sec:numerical-examples}

Inspired by papers \cite{Abate2008,Yang2018}, we apply our methods to study safety probability of a thermostatically controlled load given by the dynamics
\begin{equation}
x_{k+1} = \alpha x_{k} + (1-\alpha)(\theta-\eta R P a_k) + \omega_k,
\label{eq:TCL-model}
\end{equation}
where the state $x_k \in \mathbb{R}$ is the temperature, $a_k \in \{0,1\}$ is a binary control input, representing whether the load is on or off, and $\omega_k$ is a stochastic disturbance taking values in the uncertainty space $(\Omega,\mathcal{F},\mathbb{P})$. The parameters of \eqref{eq:TCL-model} are given by $\alpha = \exp(h/CR)$, where $R = 2 \degree \mathrm{C}/\mathrm{kW}$, $C = 2 \mathrm{kWh}/\degree \mathrm{C}$, $\theta = 32 \degree \mathrm{C}$, $h = 5/60~\mathrm{hour}$, $P = 14 \mathrm{kW}$, and $\eta = 0.7$. Our goal is keep the temperature within the range $\mathcal{S} = [19 \degree \mathrm{C},22 \degree C]$ for $90$ minutes.

Our goal is to compute a control policy purely on available sampled trajectories for the model \eqref{eq:TCL-model} and without solving an optimization problem at each iteration. To this end, we let \changesroundOne{$\{\widehat{x}_{(i)},\widehat{a}_{(i)},\widehat{x}_{(i)}^+\}_{i = 1}^m$} be a collection of observed transitions from the model,  where the pair \changesroundOne{$\{\widehat{x}_{(i)},\widehat{a}_{(i)}\}$} is the set of random chosen points in the set $[19,22]\times \{0,1\}$ and \changesroundOne{$\widehat{x}_{(i)}^+$} represents the observed transitions from such a state-input pair. We then solve the dynamic programming recursion given in \eqref{eq:safety_DP} by partitioning the state space uniformly from $18\degree \mathrm{C}$ to $23\degree \mathrm{C}$ with 35 points, and using $7000$ data points to estimate the conditional kernel mean embedding map (see Theorem \ref{theo:conditional-mean-embedding}) and to compute the norm of the value function, as shown in \eqref{eq:regression-RKHS-norm}. We choose $\lambda = 200$ as the regularisation parameter, and use the kernel functions \changesroundOne{$k_{\XX}: \XX \times \XX \rightarrow \mathbb{R}$ and $k_{\YY}:(\mathcal{X} \times \mathcal{U}) \times  (\mathcal{X} \times \mathcal{U}) \rightarrow \mathbb{R}$ given by}
\changesroundOne{
	\begin{align}
		k_{\XX}(x,x') &= e^{\gamma|x-x'|^2}, \nonumber \\
		\changesroundOne{k_{\YY}}((x,a),(x',a')) &= e^{-\gamma | x - x'|^2} +  k_1(a,a'),
		\label{eq:kernels}
	\end{align}
with $\gamma = 100$} and where $k_1(a,a') = 1 + a a' + a a' \min(a,a') - \frac{a + a'}{2} \min(a,a')^2 + \frac{1}{3}\min(a,a')^3$,  for the numerical examples. The choice of the kernel $k_1$ has shown better results for this problem when compared with a Gaussian kernel. Figure \ref{fig:kernel-versus-moment} shows the obtained value function for different values of the radius $\epsilon$, where we notice a decrease in the returned value function with the increase in the size of the ambiguity set. The $y$-axis represents the safety probability and the $x$-axis is the temperature; notice that the value function is zero outside the safe set $[19\degree \mathrm{C}, 22 \degree \mathrm{C}]$. We also compare the returned value function with the one obtained using the method proposed in \cite{Yang2018} (we refer to the reader to this paper for the definition of the parameters $c$ and $b$ shown in the legend).





\begin{figure}
	\centering
	\includegraphics[width=0.8\columnwidth]{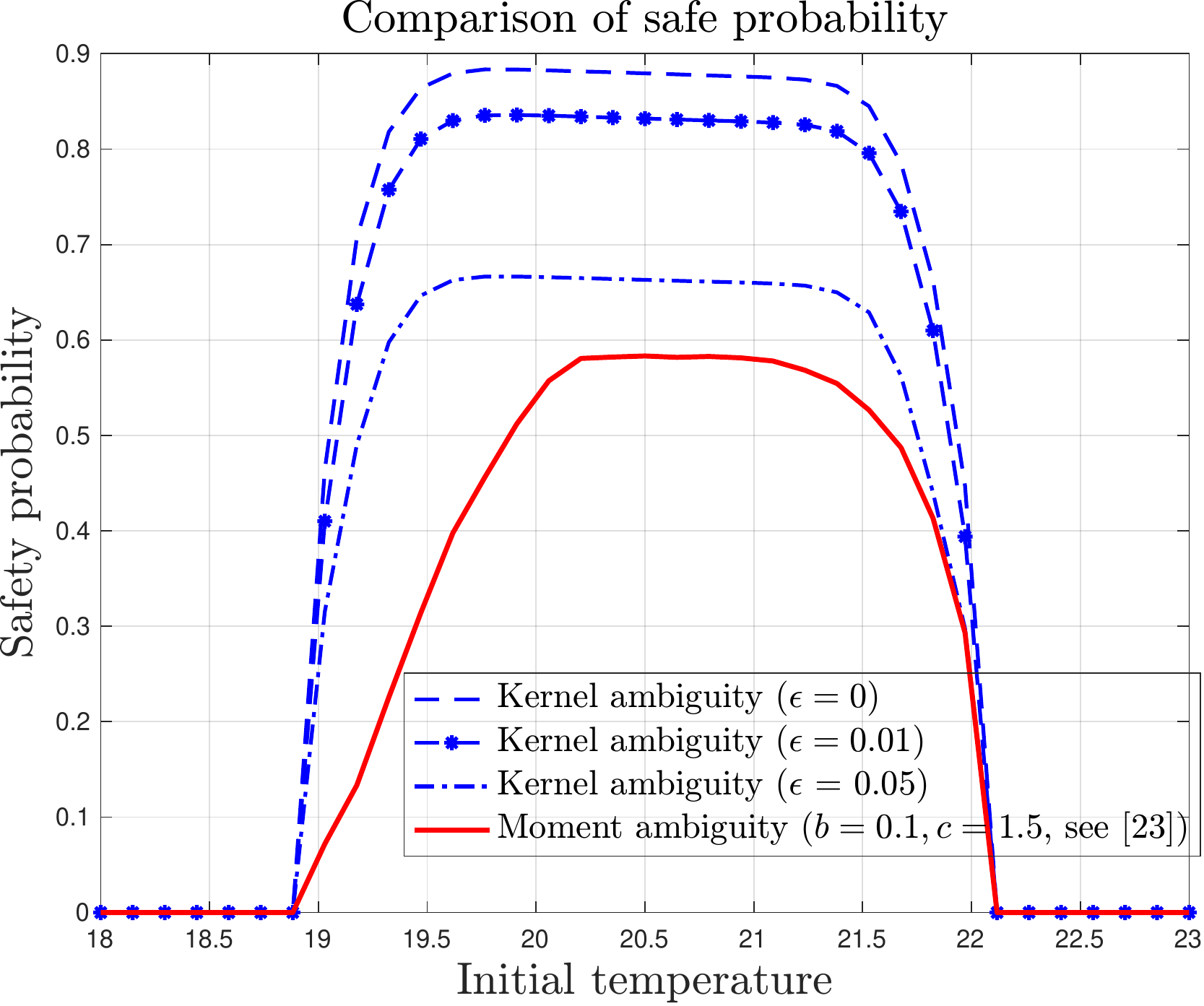}
	\caption{Solution of the dynamic programming recursion given in \eqref{eq:safety_DP} for the kernel ambiguity sets with different values of the radius (blue lines) with the kernel parameter $\gamma = 100$. We have used $7000$ state-action pairs as samples to estimate the conditional mean embedding and the norm of vectors in the RKHS. The regularisation parameter $\lambda$ is equal to $200$. The solid red line is the value function using the methods proposed in \cite{Yang2018} for $b =  0.1$ and $c = 1.5$ (as defined in \cite{Yang2018}).}
	\label{fig:kernel-versus-moment}
\end{figure}

\section{Conclusion}
\label{sec:conclusion}
We analyzed the problem of distributionally robust (safe) control of stochastic systems where the ambiguity set is defined as the set of distributions whose kernel mean embedding is within a certain distance from the empirical estimate of the conditional kernel mean embedding derived from data. We showed that there exists a non-randomized Markovian policy that is optimal and discussed how to compute the value function by leveraging strong duality associated with kernel DRO problems. Numerical results illustrate the performance of the proposed formulations and the impact of the radius of the ambiguity set. There are several possible directions for future research, including deriving efficient algorithms to perform value iteration without resorting to discretization, representing multistage state evolution using composition of conditional mean embedding operators, and performing a thorough empirical investigation on the impact of dataset size on the performance and computational complexity of the problem.

\bibliographystyle{ieeetr}
\bibliography{licio-bib-files,after-CDC}

\begin{thebibliography}{10}

\bibitem{Abate2008}
A.~Abate, M.~Prandini, J.~Lygeros, and S.~Sastry, ``Probabilistic reachability and safety for controlled discrete time stochastic hybrid systems,'' {\em Automatica}, vol.~44, no.~11, pp.~2724--2734, 2008.

\bibitem{Summers2010}
S.~Summers and J.~Lygeros, ``Verification of discrete time stochastic hybrid systems: A stochastic reach-avoid decision problem,'' {\em Automatica}, vol.~46, no.~12, pp.~1951--1961, 2010.

\bibitem{Persis2020}
C.~Persis and P.~Tesi, ``Formulas for data-driven control: stabilization, optimality, and robustness,'' {\em {IEEE} Transactions on Automatic Control}, vol.~65, no.~3, pp.~909--924, 2020.

\bibitem{Berberich2021}
J.~Berberich, J.~Kohler, M.~A. Muller, and F.~Allgower, ``Data-driven model predictive control with stability and robustness guarantees,'' {\em {IEEE} Transactions on Automatic Control}, vol.~66, no.~4, pp.~1702--1717, 2021.

\bibitem{Shapiro2017}
A.~Shapiro, ``Distributionally robust stochastic programming,'' {\em {SIAM} Journal on Optimization}, vol.~27, no.~4, pp.~2258--2275, 2017.

\bibitem{MohajerinEsfahani2018}
P.~Mohajerin~Esfahani and D.~Kuhn, ``Data-driven distributionally robust optimization using the wasserstein metric: performance guarantees and tractable reformulations,'' {\em Mathematical Programming}, vol.~171, no.~1, pp.~115--166, 2018.

\bibitem{HCL19}
A.~Hota, A.~Cherukuri, and J.~Lygeros, ``Data-driven chance constrained optimization under wasserstein ambiguity sets,'' in {\em 2019 American Control Conference ({ACC})}, {IEEE}, 2019.

\bibitem{ALCD22}
L.~Aolaritei, N.~Lanzetti, H.~Chen, and F.~Dörfler, ``Distributional uncertainty propagation via optimal transport,'' {\em arXiv}, 2023.

\bibitem{GonzalezTrejo2002}
J.~Gonz{\'a}lez-Trejo, O.~Hern{\'a}ndez-Lerma, and L.~Hoyos-Reyes, ``Minimax control of discrete-time stochastic systems,'' {\em SIAM Journal on Control and Optimization}, vol.~41, no.~5, pp.~1626--1659, 2002.

\bibitem{Ding2013}
J.~Ding, M.~Kamgarpour, S.~Summers, A.~Abate, J.~Lygeros, and C.~Tomlin, ``A stochastic games framework for verification and control of discrete time stochastic hybrid systems,'' {\em Automatica}, vol.~49, no.~9, pp.~2665--2674, 2013.

\bibitem{Muandet2017}
K.~Muandet, K.~Fukumizu, B.~Sriperumbudur, and B.~Sch{\"o}lkopf, ``Kernel mean embedding of distributions: A review and beyond,'' {\em Foundations and Trends{\textregistered} in Machine Learning}, vol.~10, no.~1-2, pp.~1--141, 2017.

\bibitem{ZJDS21}
J.~Zhu, W.~Jitkrittum, M.~Diehl, and B.~Sch{\"o}lkopf, ``Kernel distributionally robust optimization: Generalized duality theorem and stochastic approximation,'' in {\em International Conference on Artificial Intelligence and Statistics}, pp.~280--288, PMLR, 2021.

\bibitem{SGSS07}
A.~Smola, A.~Gretton, L.~Song, and B.~Sch{\"o}lkopf, ``A {H}ilbert space embedding for distributions,'' in {\em International Conference on Algorithmic Learning Theory}, pp.~13--31, Springer, 2007.

\bibitem{Fukumizu2013}
K.~Fukumizu, L.~Song, and A.~Gretton, ``Kernel {B}ayes' rule: {B}ayesian inference with positive definite kernels,'' {\em The Journal of Machine Learning Research}, vol.~14, no.~1, pp.~3753--3783, 2013.

\bibitem{BGG13}
B.~Boots, A.~Gretton, and G.~J. Gordon, ``Hilbert space embeddings of predictive state representations,'' in {\em Proceedings of the Twenty-Ninth Conference on Uncertainty in Artificial Intelligence}, pp.~92--101, 2013.

\bibitem{Thorpe2019}
A.~Thorpe and M.~Oishi, ``Model-free stochastic reachability using kernel distribution embeddings,'' {\em IEEE Control Systems Letters}, vol.~4, no.~2, pp.~512--517, 2019.

\bibitem{Thorpe_2022}
A.~Thorpe and M.~Oishi, ``{SOCKS}: A stochastic optimal control and reachability toolbox using kernel methods,'' in {\em 25th {ACM} International Conference on Hybrid Systems: Computation and Control}, {ACM}, 2022.

\bibitem{CKA22}
Y.~Chen, J.~Kim, and J.~Anderson, ``Distributionally robust decision making leveraging conditional distributions,'' in {\em 2022 IEEE 61st Conference on Decision and Control (CDC)}, pp.~5652--5659, IEEE, 2022.

\bibitem{Pillonetto2011}
G.~Pillonetto, M.~Quang, and A.~Chiuso, ``A new kernel-based approach for nonlinear system identification,'' {\em {IEEE} Transactions on Automatic Control}, vol.~56, no.~12, pp.~2825--2840, 2011.

\bibitem{GLB12}
S.~Gr{\"u}new{\"a}lder, G.~Lever, L.~Baldassarre, S.~Patterson, A.~Gretton, and M.~Pontil, ``Conditional mean embeddings as regressors,'' in {\em Proceedings of the 29th International Coference on International Conference on Machine Learning}, pp.~1803--1810, 2012.

\bibitem{Thorpe2022}
A.~Thorpe, K.~Ortiz, and M.~Oishi, ``State-based confidence bounds for data-driven stochastic reachability using {H}ilbert space embeddings,'' {\em Automatica}, vol.~138, p.~110146, 2022.

\bibitem{Yang2020}
I.~Yang, ``Wasserstein distributionally robust stochastic control: a data-driven approach,'' {\em IEEE Transactions on Automatic Control}, vol.~66, no.~8, pp.~3863--3870, 2020.

\bibitem{Yang2018}
I.~Yang, ``A dynamic game approach to distributionally robust safety specifications for stochastic systems,'' {\em Automatica}, vol.~94, pp.~94--101, 2018.

\bibitem{Schuurmans2023}
M.~Schuurmans and P.~Patrinos, ``A general framework for learning-based distributionally robust {MPC} of {M}arkov jump systems,'' {\em IEEE Transactions on Automatic Control}, vol.~68, no.~5, pp.~2950--2965, 2023.

\bibitem{FL22}
M.~Fochesato and J.~Lygeros, ``Data-driven distributionally robust bounds for stochastic model predictive control,'' in {\em 2022 IEEE 61st Conference on Decision and Control (CDC)}, pp.~3611--3616, IEEE, 2022.

\bibitem{Luo2020}
F.~Luo and S.~Mehrotra, ``Distributionally robust optimization with decision dependent ambiguity sets,'' {\em Optimization Letters}, vol.~14, pp.~2565--2594, 2020.

\bibitem{Noyan2022}
N.~Noyan, G.~Rudolf, and M.~Lejeune, ``Distributionally robust optimization under a decision-dependent ambiguity set with applications to machine scheduling and humanitarian logistics,'' {\em INFORMS Journal on Computing}, vol.~34, no.~2, pp.~729--751, 2022.

\bibitem{Salamon16}
D.~Salamon, {\em Measure and Integration}.
\newblock European Mathematical Society, 2016, 2016.

\bibitem{Brezis11}
H.~Brezis, {\em Functional Analysis, Sobolev Spaces and Partial Differential Equations}.
\newblock Springer, 2011.

\bibitem{Panchenko19}
D.~Panchenko, {\em Lecture Notes on Probability Theory}.
\newblock 2019.

\bibitem{Nemmour2022}
Y.~Nemmour, H.~Kremer, B.~Sch{\"o}lkopf, and J.~Zhu, ``Maximum mean discrepancy distributionally robust nonlinear chance-constrained optimization with finite-sample guarantee,'' in {\em IEEE Conference on Decision and Control {(CDC)}}, pp.~5660--5667, 2022.

\bibitem{Baker1973}
C.~Baker, ``Joint measures and cross-covariance operators,'' {\em Transactions of the American Mathematical Society}, vol.~186, pp.~273--273, 1973.

\bibitem{Bertsekas.Shreve78}
D.~Bertsekas and S.~Shreve, {\em Stochastic Optimal Control: The Discrete-time Case}.
\newblock Athena Scientific, 1978.

\bibitem{Fukumizu2004}
K.~Fukumizu, F.~Bach, and M.~Jordan, ``Dimensionality reduction for supervised learning with reproducing kernel {H}ilbert spaces,'' {\em Journal of machine learning research}, 2004.

\bibitem{SHSF09}
L.~Song, J.~Huang, A.~Smola, and K.~Fukumizu, ``{H}ilbert space embeddings of conditional distributions with applications to dynamical systems,'' in {\em Proceedings of the 26th Annual International Conference on Machine Learning}, pp.~961--968, 2009.

\bibitem{GLB12b}
S.~Grunewalder, G.~Lever, L.~Baldassarre, M.~Pontil, and A.~Gretton, ``Modelling transition dynamics in {MDP}s with {RKHS} embeddings,'' in {\em International Conference on Machine Learning {(ICML)}}, pp.~1603--1610, 2012.

\bibitem{Micchelli2005}
C.~Micchelli and M.~Pontil, ``On learning vector-valued functions,'' {\em Neural computation}, vol.~17, no.~1, pp.~177--204, 2005.

\bibitem{Durrett10}
R.~Durrett, {\em Probability: Theory and Examples}.
\newblock Cambridge University Press, 2010.

\bibitem{Munkres74}
J.~Munkres, {\em Topology: A First Course}.
\newblock Prentice Hall, 1974.

\bibitem{Aubin.Frankowska09}
J.-P. Aubin and H.~Frankowska, {\em Set-valued analysis}.
\newblock Springer Science \& Business Media, 2009.

\bibitem{Wood2021}
K.~Wood, G.~Bianchin, and E.~Dall’Anese, ``Online projected gradient descent for stochastic optimization with decision-dependent distributions,'' {\em IEEE Control Systems Letters}, vol.~6, pp.~1646--1651, 2021.

\end{thebibliography}

\end{document}